\documentclass{article}

\usepackage{setspace}
\usepackage{color}

\usepackage{amsfonts}
\usepackage{amsmath}
\usepackage{amsthm}
\usepackage{amssymb}

\newcommand{\br}{\mathbb R}

\newcommand{\nc}{\newcommand}
\newtheorem{lemma}{Lemma}[section]
\newtheorem{theorem}{Theorem}[section]
\newtheorem{problem}{Problem}[section]
\newtheorem{definition}{Definition}[section]

\newtheorem{proposition}{Proposition}[section]

\nc{\cD}{{\cal D}}
\nc{\cP}{{\cal P}}
\nc{\cR}{{\cal R}}
\nc{\e}{\varepsilon}
\nc{\Om}{\Omega}
\nc{\om}{\omega}
\nc{\aal}{\alpha}
\nc{\tow}{\rightharpoonup}
\nc{\nin}{\in \hs{-.35}/\,}
\nc{\np}{\newpage}
\nc{\g}{\gamma}
\nc{\IN}{I \hs{-.15} N}
\nc{\IR}{I \hs{-.14} R}
\nc{\IK}{I \hs{-.14} K}
\nc{\hs}[1]{\hspace{#1cm}}

\title{\bf On the existence of an exponential attractor for a planar shear flow with Tresca's friction condition}

\author{Grzegorz {\L}ukaszewicz 
    \thanks{E-mail : glukasz@mimuw.edu.pl, Tel.: +48 22 55 44 562}
\thanks{This research was supported by Polish Government Grant N N201 547638} }

{\large\date{}}

\begin{document}
\maketitle
\begin{center}

 {\small
 University of Warsaw, Mathematics Department,
     ul.Banacha 2, 02-957 Warsaw, Poland}
\end{center}
\normalsize

\begin{abstract}
\noindent We consider a two-dimensional nonstationary Navier-Stokes shear flow with a sub\-differential boundary condition on a part of the boundary of the flow domain, namely, with a boundary driving subject to the Tresca law. There exists a unique global in time solution of the considered problem which is governed by a variational inequality.
Our aim is to prove the existence of a global attractor of a finite fractional dimension and of an exponential attractor for the associated semigroup. We use the method of $l$-trajectories. This research is motivated by a problem from lubrication
theory.
\end{abstract}

\vspace{0.2cm}


\vspace{0.2cm} \noindent{\it 1991 Mathematics Subject
Classification:} 76D05, 76F10, 76F20 \vspace{0.2cm}
 \renewcommand{\theequation}{\arabic{section}.\arabic{equation}}
 \setcounter{equation}{0}
 \section{Introduction}
 
Remarking on future directions of research in the field of contact mechanics, in their recent book \cite{millor-sofonea-telega-2010}, 
the authors wrote: "The infinite-dimensional dynamical systems approach to contact problems is virtually nonexistent. (...) This topic certainly deserves further consideration". 

From the mathematical point of view a considerable difficulty in analysing problems of contact mechanics, and dynamical problems in particular, comes from the presence of involved boundary constraints which are often modelled by boundary conditions of a dissipative subdifferential type and lead to a formulation of the considered problem in terms of a variational or hemivariational inequality with, frequently, nondifferentiable boundary functionals. 

Our aim in this paper is to contribute to this topic by an examination of the large time behaviour of solutions of a problem coming from the theory of lubrication. 

We study the problem of existence of the global attractor of a finite fractal dimension and of an exponential attractor for a class of two-dimensional turbulent {\it boundary driven} flows subject to the Tresca law  which naturally appears in lubrication theory. Existence of such attractors strongly suggest that the time asymptotics of the considered flow can be described by a finite number of parameters and then treated numerically 
\cite{robinson-2011-dim, robinson-2001-infty}. We study the problem in its weak formulation given in terms of an evolutionary variational inequality with a nondifferentiable boundary functional. This situation produces an obstacle for applying directly the classical methods, presented e.g., in monographs 
\cite{robinson-2001-infty, chep-vish-2002, cholewa-dlotko 2000, hale-1988, temam-infty}, to prove that the fractal dimension of the global attractor is finite. Instead, we apply the powerful method of $l$-trajectories, introduced in \cite{malek-necas-1996, malek-prazak-2002} which
we use further to prove the existence of an exponential attractor. The method of $l$-trajectories helps to prove the existence of an exponential attractor for a considerably large class of nonlinear problems, in particular that with lack of good regularity properties 
(c.f., e.g., \cite{Feireisl-Prazak-2010, Miranville-Zelik-2008, segatti-zelik-2010} and references therein). 

The problem we consider is as follows. The flow of an incompressible fluid in a two-dimensional domain $\Omega$
is described by the equation of motion
\begin{equation} \label{e2.1}
 u_{t} -\nu\Delta u + (u\cdot\nabla) u +\nabla p = 0 \quad {\rm in} \quad \Omega
\end{equation}
\noindent and the incompressibility condition
\begin{equation} \label{e2.2}
{\rm div} \, u = 0 \quad {\rm in} \quad \Omega.
\end{equation}
To define the domain $\Omega$ of the flow, let $\Omega_{\infty}$ be the channel,
$$\Omega_{\infty} =\{x=(x_1,x_2): -\infty < x_1 < \infty,\,\,\,
          0 < x_2 < h(x_{1})\}, $$
where $h$ is a positive function,  smooth, and L-periodic in $x_1$.
Then we set
  $$\Omega =   \{x=(x_1,x_2): 0 < x_1 < L, \,\,\, 0 < x_2 < h(x_{1})\} $$
and
  $\partial\Omega=\bar{\Gamma}_{0}\cup\bar{\Gamma}_{L}\cup\bar{\Gamma}_{1}$,
where $\Gamma_{0}$ and $\Gamma_{1}$ are the bottom and the top,
and $\Gamma_{L}$ is the lateral part of the boundary of $\Omega$.

We are interested in solutions of (\ref{e2.1})-(\ref{e2.2})
in $\Omega$ which are L-periodic with respect to  $x_1$.
 We assume that
\begin{equation} \label{e2.3}
 u=0 \quad  {\rm at} \quad \Gamma_{1}.
 \end{equation}

\noindent Moreover, we assume that there is  no flux condition across $\Gamma_{0}$ so that
the normal component of the velocity on $\Gamma_{0}$  satisfies
\begin{equation} \label{e2.4}
u\cdot n = 0 \quad  {\rm at} \quad \Gamma_{0},
\end{equation}

\noindent and that the tangential component of the velocity $u_{\eta}$
on $\Gamma_{0}$ is unknown and satisfies the Tresca friction law with a constant and positive maximal
friction coefficient $k$. This means that, c.f., e.g., \cite{millor-sofonea-telega-2010, Duv72},

\begin{equation} \label{e2.5}
\left.
\begin{array}{l}
|\sigma_{\eta}(u , p)| \leq k \\ \\
|\sigma_{\eta}(u , p)|<k  \Rightarrow u_{\eta} =U_{0} e_{1}\\ \\
|\sigma_{\eta}(u , p)|=k  \Rightarrow \exists \lambda\geq 0 \mbox{ such that }
 u_{\eta}= U_{0} e_{1}-\lambda\sigma_{\eta}(u , p)
\end{array}
\right\}\quad {\rm at} \quad \Gamma_{0}
\end{equation}

\noindent where $\sigma_{\eta}$ is the tangential component of the stress tensor on $\Gamma_{0}$ and
$U_{0} e_{1}=(U_0, 0)$, $U_0\in \br$,
is the velocity of the lower surface producing the driving force of the flow.

If $n=(n_{1} ,  n_{2})$ is the unit outward normal to $\Gamma_0$,  and $\eta=(\eta_{1} ,  \eta_{2})$ is the unit
tangent vector to $\Gamma_0$ then we have
\begin{eqnarray}\label{eT}
      \sigma_{\eta}(u , p) = \sigma(u , p)\cdot n -
 ((\sigma(u , p)\cdot n)\cdot n ) n,
\end{eqnarray}
where
 $\sigma(u,p)=(\sigma_{ij}(u , p)) =( -p \delta_{ij}  +\nu \left(u_{i, j} + u_{j , i}\right))$
is the stress tensor.
\noindent Finally, the initial condition for the velocity field is
 \begin{eqnarray*} \label{e2.6}
  u(x , 0) = u_{0}(x)   \quad
  {\rm for} \quad x\in\Omega.
 \end{eqnarray*}

\noindent  The problem is motivated by the examination of a certain two-dimensional flow in an infinite
(rectified) journal bearing $\Omega\times(-\infty,+\infty)$, where
$\Gamma_1\times(-\infty,+\infty)$ represents the outer cylinder,
and $\Gamma_0\times(-\infty,+\infty)$ represents the inner,
rotating cylinder. In the lubrication problems the gap $h$ between
cylinders is never constant. We can assume that the rectification
does not change the equations as the gap between cylinders is very
small with respect to their radii.

The knowledge or the judicious choice of the boundary conditions on the fluid-solid interface is of particular interest in lubrication area which is concerned with thin film flow behaviour. The boundary conditions to be employed are determined by numerous physical parameters characterizing, for example, surface roughness and rheological properties of the fluid.

The widely used no-slip condition when the fluid has the same velocity as surrounding solid boundary is not respected if the shear rate becomes too high (no-slip condition is induced by chemical bounds between the lubricant and the surrounding surfaces and by the action of the normal stresses, which are linked to the pressure inside the flow; on the contrary, when tangential stressses are high they can destroy the chemical bounds and induce slip phenomenon). We can model such situation by a transposition of the well-known friction laws between two solids \cite{millor-sofonea-telega-2010} to the fluid-solid interface.

The system of equations (\ref{e2.1})-(\ref{e2.2}) with boundary conditions: (\ref{e2.3}) at $\Gamma_1$
for $h=const$ and $u=const$ on $\Gamma_{0}$, instead of
(\ref{e2.4})-(\ref{e2.5}),  was intensively studied in
several contexts, some of them mentioned in the introduction
of \cite{mb-gl-cambridge-2009}.
The autonomous case with $h\neq const$
and with $u=const$ on $\Gamma_{0}$
 was considered in \cite{BoLu1-04, BoLu2-04}.
 See also \cite{BoLuR-05} where the case $h\neq const$,
 $u=U(t)e_{1}$ on $\Gamma_{0}$, was considered.
The dynamical problem, important for applications, we consider
in this paper has been studied earlier in \cite{mb-gl-parabolic-2008} in the nonautonomous case for which the existence of a pullback attractor was established with the use of a method that, however, did not guarantee the finite dimensionality of the pullback attractor (or the global attractor in the reduced autonomous case).

To establish the existence of the global attractor of a finite fractal dimension we use the method 
of $l$-trajectories as presented in \cite{malek-prazak-2002}.
This method appears very useful when one deals with variational inequalities, cf., \cite{segatti-zelik-2010},
as it overcomes obstacles coming from the usual methods. One needs neither
compactness of the dynamics which results from the second energy
inequality nor asymptotic compactness, cf., i.e., \cite{temam-infty, BoLuR-05}, which results from the
energy equation. In the case of variational inequalities it is sometimes
not possible to get the second energy inequality and the differentiability of the associated semigroup due to
the presence of nondifferentiable boundary functionals. On the other hand, we do not
have an energy equation to prove the asymptotic compactness.

While there are other methods to establish the existence of the global attractor where the problem of the lack of regularity appears, that, e.g., based on the notion of the Kuratowski measure of noncompactness of bounded sets, where we do not need even the continuity of the semigroup associated with a given dynamical problem, cf., e.g., 
\cite{zhong-2006}, and also \cite{mb-gl-parabolic-2008}, where the nonautonomous version of the problem considered in this paper was studied, the problem of a finite dimensionality of the attractor is more involved, cf. also 
\cite{mb-gl-cambridge-2009}. 

The method of $l$-trajectories allows to prove the existence of an even more desirable object, called exponential attractor, for many problems for which there exists a finite dimensional global attractor \cite{Miranville-Zelik-2008}.  An exponential attractor is a compact subset of the phase space which is positively invariant, has finite fractal dimension, and attracts uniformly bounded sets at an exponential rate.
It contains the global attractor and thus its existence implies the finite dimensionality of the global attractor itself. Its crucial property is an exponential rate of attraction of solution trajectories \cite{Feireisl-Prazak-2010, Miranville-Zelik-2008}. The proof of the existence of an exponential attractor requires the solution to be regular enough to ensure the H\"older continuity of the semigroup in the time variable \cite{malek-prazak-2002}.
We establish this property by providing additional a priori estimates of solutions. 

Our plan is as follows.  In Section \ref{secvf} we homogenize first the boundary condition (\ref{e2.5}) by
a smooth background flow (a simple version of the Hopf construction, cf., e.g., \cite{mb-gl-parabolic-2008}) and then we present a variational formulation of the homogenized problem. In Section~\ref{existence-sol} we recall briefly the proof of the existence and uniqueness of a global in time solution of our problem and obtain some estimates of the solutions. Section \ref{l-trajectories} is devoted to a presentation of the main definitions and elements of the theory of infinite dimensional dynamical systems we use, in particular, of the method of $l$-trajectories.
In Section \ref{global-attractor} we prove the existence of the global attractor of a finite fractal dimension.
At last, in Section \ref{exponential} we prove the existence of an exponential attractor and in Section \ref{conclusions} we provide some final comments.

\renewcommand{\theequation}{\arabic{section}.\arabic{equation}}
\setcounter{equation}{0}
\section{Variational formulation of the problem}\label{secvf}
First, we homogenize the boundary condition (\ref{e2.5}). To this end let
 \begin{equation} \label{e2.7}
 u(x_{1}, x_{2} , t) = U(x_{2})e_1 + v(x_{1} , x_{2} , t)
 \end{equation}
 with
 \begin{equation} \label{e2.8}
 U(0)=U_0, \quad U(h(x_{1}))=0, \quad x_1\in (0 , L).
 \end{equation}

\noindent The new vector field $v$ is L-periodic in   $x_{1}$ and satisfies the equation of motion
\begin{equation} \label{e2.9}
      v_{t} -\nu\Delta v + (v\cdot\nabla) v  + \nabla p =  G(v)
\end{equation}
\noindent with
 \begin{equation*} \label{G}
      G(v)= - Uv,_{x_1} - (v)_2 \, U,_{x_2}e_1 +  \nu U,_{x_2 x_2}e_1 
 \end{equation*}
\noindent where by $(v)_2$ we denoted the second component of $v$.
As $ {\rm div}(U e_{1})=0$ we get
     \begin{equation} \label{e2.10}
 {\rm div} \, v = 0 \quad {\rm in} \quad \Omega.
    \end{equation}

\noindent From (\ref{e2.7})-(\ref{e2.8}) we obtain

\begin{equation} \label{e2.11}
  v = 0,      \quad {\rm on} \quad \Gamma_{1},
\end{equation}
and
\begin{equation} \label{e2.12}
  v\cdot n = 0,      \quad {\rm on} \quad \Gamma_{0}.
     \end{equation}
Moreover, we have,
\begin{equation*}
 \sigma_\eta(v,p)=\sigma_\eta(u,p) + (\nu\frac{\partial U(x_2)}{\partial x_2}|_{x_2=0}, 0).
\end{equation*}
Since we can define the extension $U$ in such a way that
      $$\frac{\partial U(x_2)}{\partial x_2}|_{x_2=0} = 0$$
\noindent the Tresca condition (\ref{e2.5}) transforms to
  \begin{equation} \label{e2.13}
\left.
\begin{array}{l}
       |\sigma_{\eta}(v , p)| \leq k \\ \\
       |\sigma_{\eta}(v , p)|<k  \Rightarrow v_{\eta} =0\\ \\
       |\sigma_{\eta}(v , p)|=k  \Rightarrow \exists \lambda\geq 0 \mbox{ such that }
 v_{\eta}= -\lambda\sigma_{\eta}(v , p)
\end{array}
\right\}\quad {\rm at} \quad \Gamma_{0}
\end{equation}
\noindent Finally, the initial condition becomes
 \begin{equation} \label{e2.14}
      v(x , 0)= v_{0}(x)= u_{0}(x) - U(x_{2})e_{1}.
 \end{equation}
The Tresca condition (\ref{e2.13}) is a particular case of an important in contact mechanics class of 
subdifferential boundary conditions of the form, cf., e.g. \cite{papa-1985},
\begin{eqnarray} \label{subdiff1}
     \varphi(\Theta) - \varphi(v) \geq -\sigma n (\Theta-v) \quad {\rm at} \,\,\,\, \Gamma_0,
\end{eqnarray}
where $\sigma n$ is the Cauchy stress vector and $\Theta$ belongs to a certain set of admissible functions. For 
    $\varphi(v) = k|v_\eta|$ the last condition is equivalent to (\ref{e2.13}).
    
Now we can introduce the variational formulation of the  homogenized problem  (\ref{e2.9})-(\ref{e2.14}). Then, for the convenience of the readers, we describe the relations between the classical and the weak formulations.

We begin with some basic definitions of the paper.

Let
  \begin{eqnarray*}
      \tilde{V}&=& \{v \in {\cal C}^{\infty}(\Omega)^{2}:\,\, {\rm div} \, v = 0 \,\,\mbox{ in }\,\,\Om,
      \,\, v \, \mbox{is L-periodic in }\, x_1 ,\,\,
      \nonumber\\
      &&\qquad  \,\,\, v=0 \, \mbox{ at }
      \Gamma_1, \quad v\cdot n= 0 \, \mbox{ at } \Gamma_0 \}
  \end{eqnarray*}
and
$$
V =   {\rm closure \,\,  of\,\, } \tilde{V} \,\, {\rm  in }  \,\, H^{1}(\Omega)^2, \qquad
H = {\rm  closure \,\,  of   \, } \tilde{V}  \,\, {\rm in }\,\,  L^{2}(\Omega)^2.
$$
\noindent We define scalar products in
   $H$
and
   $V$,
respectively, by
   $$(u \, ,\,   v) = \int_\Omega u(x)v(x)dx \quad
    \mbox{ and }\quad (\nabla u \, ,\,  \nabla v) $$
 and their associated norms by
  $$|v| = (v, v)^{\frac{1}{2}} \quad {\rm and} \quad
    \|v\| = (\nabla v \, ,\,  \nabla v)^{\frac{1}{2}}.$$
\noindent Let, for $u, v$ and $w$ in $V$
  $$a(u \, ,\,  v) = (\nabla u \, ,\,  \nabla v) \quad {\rm and}  \quad b(u \, ,\,  v \, ,\,  w) = ((u\cdot\nabla) v \, ,\,  w).$$
\noindent In the end, let us define the functional $j$ on $V$ by
\begin{eqnarray*}\label{e2.15}
 j(u) =\int_{\Gamma_{0}} k |u(x_{1} , 0) |dx_{1}.
 \end{eqnarray*}

\noindent The variational formulation of the homogenized problem
 {\rm (\ref{e2.9})-(\ref{e2.14})} is as follows.
\begin{problem} \label{pr.2.1} Given $v_0\in H$, find $v:(0,\infty) \to H$ such that:

\noindent
 (i) for all $T > 0$,
     $$v \in {\cal C}([0, T]; H)\cap L^{2}(0 , T ; V),
     \qquad with \quad v_{t}\in L^{2}(0 , T ; V') $$
where $V'$ is the dual space to $V$.

\noindent (ii) for all $\Theta$ in $V$, all $T>0$, and for almost all $t$ in the interval $[0, T]$, the following variational inequality holds
\begin{eqnarray} \label{eqn:er2.16}
     \langle v_{t}(t),\Theta-v(t) \rangle\, + \,\nu a(v(t), \Theta- v(t)) \,
            &+& \, b(v(t),v(t), \Theta - v(t)) \,
        \\ \nonumber \\
      &+& j(\Theta) - j( v(t)) \,\geq \, ({\cal L}(v(t)),\Theta- v(t)) \nonumber
\end{eqnarray}

\noindent (iii) the initial condition
\begin{eqnarray}\label{eqn:er2.17}
 v(x , 0) = v_{0} (x)
\end{eqnarray}
holds.

In (\ref{eqn:er2.16}) the functional ${\cal L}(v(t))$ is defined for almost all $t\geq 0$ by,
\begin{eqnarray*}\label{er2.4}
     ({\cal L}(v(t)), \Theta) = - \nu a(\xi ,\Theta) - b(\xi, v(t),\Theta) - b(v(t) , \xi , \Theta),
\end{eqnarray*}
\noindent where $\xi = Ue_{1}$ is a suitable smooth background flow.
\end{problem}
We have the following relations between classical and weak formulations.
 \begin{proposition}
Every classical solution of Problem {\rm (\ref{e2.9})-(\ref{e2.14})} is also a solution of 
Problem {\rm \ref{pr.2.1}}. On the other hand, every solution of Problem {\rm \ref{pr.2.1}} which is smooth enough 
is also a classical solution of Problem {\rm (\ref{e2.9})-(\ref{e2.14})}.
 \end{proposition}
\begin{proof}
Let $v$ be a classical solution of Problem (\ref{e2.9})-(\ref{e2.14}). As it is (by assumption)
 sufficiently regular, we have to check only (\ref{eqn:er2.16}).
 Remark first that (\ref{e2.9}) can be written as
\begin{eqnarray}\label{e2.9bis}
     v_{t}- {\rm Div}\,\sigma(v , p)+ (v\cdot\nabla) v = G(v(t)).
\end{eqnarray}
Let $\Theta \in V$. Multiplying (\ref{e2.9bis}) by $\Theta- v(t)$
and using Green's formula we obtain
\begin{eqnarray}\label{eq2.6}
\int_{\Omega}v_{t}(\Theta-v(t))dx
&+& \int_{\Omega}\sigma_{ij}(v , p) (\Theta- v(t))_{i, j}dx
+ b(v(t) \, , \,v(t) \, , \, \Theta-v(t))
\nonumber\\
&=&\int_{\partial\Omega} \sigma_{ij}(v , p)n_{j} (\Theta- v(t))_{i}
+ \int_{\Omega}G(v(t))(\Theta-v(t))dx 
\end{eqnarray}
for $t\in (0,T)$. As $v(t)$ and $\Theta$ are in $V$, after some calculations we obtain
\begin{eqnarray}\label{a1}
\int_{\Omega} \sigma_{ij}(v , p) (\Theta- v(t))_{i, j}dx
= \nu a(v(t) \, , \, \Theta- v(t)).
\end{eqnarray}
By (\ref{subdiff1}) with $\varphi(v) = k|v_\eta|$ and taking into account the boundary conditions we get
\begin{eqnarray}\label{eq2.8}
\int_{\partial\Omega} \sigma_{ij}(v , p)n_{j}(\Theta- v(t))_{i}
 \geq -\int_{\Gamma_{0}}k (|\Theta| - |v_{\eta}(t)|) 
\end{eqnarray}
\noindent Finally,
\begin{eqnarray} \label{eqn2.16a}
\int_{\Omega}G(v(t))(\Theta-v(t))dx = \,( {\cal L}(v(t)) \, , \, \Theta-v(t)).
\end{eqnarray}
\noindent From  (\ref{a1}), (\ref{eq2.8}) and (\ref{eqn2.16a}) we see that
(\ref{eq2.6}) yields (\ref{eqn:er2.16}), and (\ref{eqn:er2.17})
is the same as~(\ref{e2.14}).

\smallskip
\noindent  Conversely, suppose that $v$ is a sufficiently smooth solution to Problem {\rm \ref{pr.2.1}}.
We have immediately (\ref{e2.10})-(\ref{e2.12}) and (\ref{e2.14}).

Now, let $\varphi$ be in the space
$(H^{1}_{div}(\Omega))^{2}= \{\varphi\in  V : \varphi=0 \mbox{ on } \Gamma\}$. We 
take $\Theta= v(t) \pm \varphi$ in
(\ref{eqn:er2.16}) to get
\begin{eqnarray*}
\langle v_{t}(t) - \nu \Delta v(t)  + (v(t)\cdot \nabla) v(t) - G(v(t)) \, ,
\,  \varphi \rangle\, =0  \qquad \forall \varphi\in
(H^{1}_{div}(\Omega))^{2}.
\end{eqnarray*}
\noindent Thus, there exists a distribution $p(t)$ on $\Omega$ such that
\begin{eqnarray}\label{e2.9'}
v_{t}(t) - \nu \Delta v(t)  + (v(t)\cdot \nabla) v(t) - G(v(t))= \nabla p(t) \quad \mbox{in } \Omega
\end{eqnarray}
\noindent so that (\ref{e2.9}) holds. Now, we shall derive the Tresca boundary condition (\ref{e2.13})
from the weak formulation. We have
\begin{eqnarray} \label{Green3}
    \int_{\Omega} \sigma_{ij}(v , p) (\Theta- v(t))_{i, j}dx = - \int_{\Omega}{\rm Div}\sigma(v,p)(\Theta-v)dx +
    \int_{\partial\Omega}\sigma n(\Theta-v)d\Gamma.
\end{eqnarray}
Applying  (\ref{a1}) and (\ref{Green3}) to (\ref{eqn:er2.16}) we get
\begin{eqnarray*}
 \int_{\Omega} (v_{t}- {\rm Div}\,\sigma(v , p)+ (v\cdot\nabla) v - G(v(t)))(\Theta-v)dx -\int_{\partial\Omega}\sigma n(\Theta-v)d\Gamma
 \geq j(v) - j(\Theta)
\end{eqnarray*}
By (\ref{e2.9'}) we have (\ref{e2.9bis}) and so the first integral on the left hand side vanishes. Thus we obtain
condition (\ref{eq2.8}). As
\begin{eqnarray*}
\int_{\partial\Omega}
\sigma_{ij}(v , p)n_{j} (\Theta- v(t))_{i}
=\int_{\Gamma_{0}} \sigma_{\eta}(v , p)(\Theta- v_{\eta}(t))
+ \int_{\Gamma_{0}} (\sigma_{ij} n_{j}n_{i}) n_{i}(\Theta- v_{\eta}(t))_{i}
\end{eqnarray*}
and the last integral equals zero as $n_{i} (\Theta- v_{\eta}(t))_{i}=0$ on $\Gamma_{0}$, inequality (\ref{eq2.8}) can be written in the form
\begin{eqnarray} \label{tresca-int}
    \int_{\Gamma_{0}} \sigma_{\eta}(v , p)(\Theta- v_{\eta}(t)) \geq 
    -\int_{\Gamma_{0}}k (|\Theta| - |v_{\eta}(t)|), 
\end{eqnarray}
where $\Theta$ is any element of $V$. From (\ref{tresca-int}) we obtain the Tresca boundary condition (\ref{e2.13}) in an elementary way, observing that (\ref{tresca-int}) implies
\begin{eqnarray*}
    -\int_{\Gamma_{0}} \sigma_{\eta}v = \int_{\Gamma_{0}}k|v_\eta| \quad {\rm and} \quad 
    \big|\int_{\Gamma_{0}} \sigma_{\eta}\Theta \big|
    \leq \int_{\Gamma_{0}}k |\Theta|.
\end{eqnarray*}
\end{proof}

\renewcommand{\theequation}{\arabic{section}.\arabic{equation}}
\setcounter{equation}{0}
\section{Existence and uniqueness of a global in time solution} \label{existence-sol}

\noindent In this section we establish, following \cite{mb-gl-parabolic-2008}, the existence and uniqueness 
of a global in time solution for Problem \ref{pr.2.1}. First, we present two lemmas. 
\begin{lemma}\label{lemma3.1} (\cite{BoLuR-05})
There exists a smooth extension
\begin{eqnarray*} \label{eqn:er3.6}
        \xi(x_2) =  U(x_2)e_1
\end{eqnarray*}
of $U_0e_1$ from $\Gamma_0$ to $\Om$ satisfying: (\ref{e2.8}),
\begin{equation*}
  \frac{\partial U(x_2)}{\partial x_2}|_{x_2=0} = 0,
\end{equation*}
and such that
\begin{eqnarray*} \label{er3.4a}
    |b(v \,,\, \xi ,\, v)|\leq \frac{\nu}{4} \|v\|^{2}
      \quad {\rm for\,\,all}\,\,\,v\in V.
\end{eqnarray*}
Moreover,
\begin{equation*} \label{eqn:er3.24.0}
        |\xi|^{2} + |\nabla \xi|^{2} =  \int_{\Omega}|U(x_2)|^2dx_1dx_2 + \int_{\Omega}|U,_{x_2}(x_2)|^2dx_1dx_2
        \leq F,
\end{equation*}
where $F$ depends on $\nu, \Omega$, and $ U_0$.
\end{lemma}
\begin{lemma} (\cite{mb-gl-parabolic-2008}) \label{lemma3.3}
For all $v$ in $V$ we have  the Ladyzhenskaya inequality
\begin{eqnarray}\label{lady}
 \|v\|_{L^{4}(\Om)}\leq C(\Om)
  |v|^{\frac{1}{2}}\|v\|^{\frac{1}{2}}.
  \end{eqnarray}
\end{lemma}
 \begin{proof}
Let $v\in V$ and $r\in C^1((-L, L))$ such that $r= 1$ on $[0
\,,\, L]$ and $r=0$ at $x_1=-L$. Define $\varphi=r v$, and
extend $\varphi$ by $0$ to $\Om_{1}= (-L \, ,\,  L)\times (0\, , \,
h)$, where $h=\max_{0\leq x_{1}\leq L}h(x_{1})$. We obtain
 \begin{eqnarray*}
 \varphi^{2}(x_{1} , x_{2}) &=&
  2\int_{-L}^{x_{1}}  \varphi(t_{1} , x_{2})
  \frac{\partial \varphi}{\partial t_{1}}(t_{1} , x_{2}) d t_{1}
        \leq
    2\int_{-L}^{L}  |\varphi(x_{1} , x_{2})|
  \, |\frac{\partial \varphi}{\partial x_{1}}(x_{1} , x_{2})| d x_{1}
\end{eqnarray*}
and
 \begin{eqnarray*}
\varphi^{2}(x_{1} , x_{2})  = -
  2\int_{x_{2}}^{h}  \varphi(x_{1} , t_{2})
  \frac{\partial \varphi}{\partial t_{2}}(x_{1} , t_{2}) d t_{2}
 \leq
2\int_{0}^{h}  |\varphi(x_{1} , x_{2})|
  \, |\frac{\partial \varphi}{\partial x_{2}}(x_{1} , x_{2})| d
  x_{2},
\end{eqnarray*}
\noindent whence
\begin{eqnarray*}
\|\varphi\|^{4}_{L^{4}(\Om_{1})}&=&
\int_{\Om_{1}}\varphi^{2}(x_{1} , x_{2})\varphi^{2}(x_{1} , x_{2})dx_{1} dx_{2}
               \nonumber\\
   &\leq&
   \left(\int_{0}^{h}\sup_{-L\leq x_{1}\leq L}
    \varphi^{2}(x_{1} , x_{2}) dx_{2}\right)
   \left(\int_{-L}^{L}\,\,  \sup_{0\leq x_{2}\leq h}
    \varphi^{2}(x_{1} , x_{2})
 dx_{1}\right)
 \nonumber\\
 &\leq&
 4\left(\int_{0}^{h}\int_{-L}^{L}
   |\varphi|
  |\frac{\partial \varphi}{\partial x_{1}}| d x_{1}dx_{2}\right)
  \times
 \left(\int_{-L}^{L}\int_{0}^{h}
|\varphi|
  |\frac{\partial \varphi}{ \partial x_{2}}| d x_{2}d x_{1}\right).
 \end{eqnarray*}
 \noindent By the Cauchy-Schwartz inequality,
 \begin{eqnarray*}
\|\varphi\|^{4}_{L^{4}(\Om_{1})}&\leq&
 4 |\varphi|^{2}_{L^{2}(\Omega_{1})}
      |\frac{\partial \varphi}{ \partial x_{1}}|_{L^{2}(\Omega_{1})}
     |\frac{\partial \varphi}{\partial x_{2}}|_{L^{2}(\Omega_{1})}
 \nonumber\\
&\leq&
 2 |\varphi|^{2}_{L^{2}(\Omega_{1})}
\left(
  |\frac{\partial \varphi}{ \partial x_{1}}|^{2}_{L^{2}(\Omega_{1})}
+ |\frac{\partial \varphi}{\partial x_{2}}|^{2}_{L^{2}(\Omega_{1})}
 \right)
 \nonumber\\
&\leq&2 |\varphi|^{2}_{L^{2}(\Omega_{1})} |\nabla\varphi|^{2}_{L^{2}(\Omega_{1})}.
\end{eqnarray*}
 We use $|r|\leq 1$ and the Poincar\'e inequality to get
$$
\|v\|_{L^{4}(\Om)} \leq \|\varphi\|_{L^{4}(\Om_{1})},
\quad
|\varphi|_{L^{2}(\Omega_{1})}\leq 2|v|_{L^{2}(\Omega)}
     \quad \mbox{ and }\quad
|\nabla\varphi|_{L^{2}(\Omega_{1})}\leq C\|v\|_{V}$$ for some
constant $C$, whence (\ref{lady}) holds.
 \end{proof}

\begin{theorem}\label{thm3.1}
For any $v_{0}\in H$ and $U_0\in\br$ there exists a solution of Problem \ref{pr.2.1}.
\end{theorem}

\begin{proof}
We provide only the main steps of the proof as it is quite standard and, on the other hand, long. The estimates we obtain will be used further in the paper.

Observe that the functional $j$ is convex, lower semicontinuous but nondifferentiable. To overcome this difficulty we use the following approach 
(see, i.e.,  \cite{Duv72}, \cite{papa-1985}, \cite{haslinger}).
For
   $\delta>0$
let
   $j_{\delta} : V\to \br$
be a  functional defined by
   $$ \varphi \mapsto j_{\delta}(\varphi) = \frac{1}{1+\delta}\int_{\Gamma_{0}} k|\varphi|^{1+\delta} dx$$
\noindent  which  is convex, lower semicontinuous and finite on $V$.
 Moreover, for
 $v_{\delta}\tow v$ in $L^2(0,T; V)$, 
   $$\liminf_{\delta\to 0^{+}}\int_0^T j_{\delta}(v_{\delta}(t))dt \geq \int_0^T j(v(t))dt$$
and
    $$ \lim_{\delta\to 0^{+}} j_{\delta}(\varphi)= j(\varphi)$$
for all $\varphi\in V$.
The functional $j_{\delta}$ is G\^ateaux differentiable in $V$, with
   $$(j'_{\delta}(v) \, , \,\Theta ) = \int_{\Gamma_{0}}k|v|^{\delta -1}\, v \, \Theta \,dx_1, 
        \quad \Theta \in V. $$
\noindent Let us consider the following equation
\begin{eqnarray}\label{e32}
     (\frac{d v_{\delta}(t)}{dt} ,\Theta) + \nu a(v_{\delta}(t), \Theta) &+&
      b(v_{\delta}(t), v_{\delta}(t),\Theta)
      + (j'_{\delta}(v_{\delta}(t)), \Theta ) \nonumber\\
     &=&
     - \nu a(\xi, \Theta)
     -  b(\xi , v_{\delta}(t) , \Theta)- b(v_{\delta}(t) ,  \xi , \Theta)
\end{eqnarray}
\noindent with initial condition
\begin{eqnarray}\label{e320}
v_{\delta}(0)= v_{0}.
\end{eqnarray}

\noindent For $\delta >0$, we  establish an a priori estimates of $v_{\delta}$.
Since  $(j'_{\delta}(v_{\delta})  ,  v_{\delta})\geq 0$, $v_{\delta}\in V$, and
    $b(v_{\delta},v_{\delta} , v_{\delta})=
    b(\xi ,{v_{\delta}} , {v_{\delta}})= 0$
then taking $\Theta=v_{\delta}(t)$ in (\ref{e32}) we get
\begin{eqnarray*} \label{er3.3}
    \frac{1}{2}\frac{d}{dt}|{v_{\delta}(t)}|^{2}+\nu \|v_{\delta}(t)\|^{2}
     \leq -\nu a(\xi  ,  v_{\delta}(t)) - b(v_{\delta}(t) , \xi , v_{\delta}(t))
\end{eqnarray*}

\noindent In view of Lemma \ref{lemma3.1} we obtain 
\begin{equation*} \label{eqn:er3.23.0}
     \frac{1}{2}\frac{d}{dt}|{v_{\delta}(t)}|^{2} + \frac{\nu}{2}\|{v_{\delta}(t)}\|^{2} \leq \nu\|{\xi}\|^2.
\end{equation*}
\noindent We estimate the right hand side in terms of the data using Lemma \ref{lemma3.1} 
to get
\begin{equation} \label{eqn:er3.30.0}
     \frac{1}{2}\frac{d}{dt}|{v_{\delta}(t)}|^{2} + \frac{\nu}{2}\|{v_{\delta}(t)}\|^{2} \leq F.
\end{equation}
with $F=F(\nu,\Omega,U_0)$.
From (\ref{eqn:er3.30.0})  we conclude that
\begin{equation} \label{er.3.19}
     |{v_{\delta}(t)}|^{2} + \nu\int_{0}^{t}\|{v_{\delta}(s)}\|^{2} ds \leq |v(0)|^{2} + 2tF,
\end{equation}
whence
 \begin{equation}\label{319}
      v_{\delta} \mbox{ is bounded in } L^{2}(0 , T ; V)\cap L^{\infty}(0, T ; H), 
      \mbox{  independently of } \delta.
 \end{equation}
\noindent  The existence of $v_{\delta}$ satisfying (\ref{e32})-(\ref{e320})
is based on  inequality (\ref{eqn:er3.30.0}), the  Galerkin approximations, and the compactness method.
Moreover, from (\ref{er.3.19}) we can deduce that
\begin{eqnarray}\label{v'}
    \dfrac{d v_{\delta}}{dt} \quad \mbox{ is  bounded in } L^{2}(0 , T ; V').
\end{eqnarray}

\noindent From (\ref{319}) and (\ref{v'}) we conclude that there exists $v$ such that
(possibly for a subsequence)
\begin{eqnarray}\label{3.19}
    v_{\delta}\tow v \quad \mbox{ in } \quad L^{2}(0 , T ; V),\quad {\rm and}  \quad
    \dfrac{d v_{\delta}}{dt}\tow  \dfrac{d v}{dt}\quad {\rm in} \quad L^{2}(0 , T ; V').
\end{eqnarray}
  \noindent  In view of (\ref{3.19}), $v\in {\cal C}([0 \, T] ; H)$, and
\begin{eqnarray*}\label{cf}
     v_{\delta}\to v \quad \mbox{ in } \quad L^{2}(0 , T ; H) \quad \mbox{ strongly.}
\end{eqnarray*}

 \noindent We can now pass to the limit $\delta \to 0$ in  (\ref{e32})-((\ref{e320}) as in \cite{Duv72}
to obtain  the variational inequality (\ref{eqn:er2.16}) for almost every $t\in (0 , T)$.
Thus the existence of a solution of Problem~\ref{pr.2.1} is established.
\end{proof}

\begin{theorem}\label{thm3.2}
Under the hypotheses of Theorem \ref{thm3.1}, the solution $v$ of Problem \ref{pr.2.1} is unique and the map $v(\tau) \to v(t)$, for $t>\tau \geq 0$, is Lipschitz continuous in $H$.
\end{theorem}

\begin{proof}
Let $v$ and $w$ be two solutions of Problem \ref{pr.2.1}. Set $\Theta=w$ in the variational inequality for $v$, $\Theta=v$ in the variational inequality for $w$, and add thus obtained inequalities. The terms with the boundary functionals reduce and for
 $u(t)= w(t)-v(t)$ we obtain
 \begin{eqnarray*}\label{334}
  \frac{1}{2}\frac{d}{dt}|u(t)|^{2} + \nu \|u(t)\|^{2} \leq
  b(u(t) , w(t), u(t))
  + b(u(t) ,\xi , u(t)).
 \end{eqnarray*}
    \noindent By Lemma \ref{lemma3.1} and the Ladyzhenskaya inequality (\ref{lady}) we obtain
\begin{equation} \label{eqn:er2.14e}
           \frac{d}{dt}|{u(t)}|^{2} + \frac{\nu}{2}
            \|{u(t)}\|^{2} \leq \frac{2}{\nu}C(\Om)^4\|w(t)\|^2 |u(t)|^2,
\end{equation}
\noindent and in view of the Poincar\'e inequality  we conclude
    \begin{equation*} \label{eqn:er2.14a}
            \frac{d}{dt}|{u(t)}|^{2} + \frac{\sigma}{2}
            |{u(t)}|^{2} \leq   \frac{2}{\nu}C(\Om)^4\|w(t)\|^2 |u(t)|^2.
    \end{equation*}
Using again the Gronwall lemma, we obtain
    \begin{equation} \label{eqn:er2.15a}
           |{u(t)}|^{2} \leq
            |{u(\tau)}|^{2} \exp\{-\int_\tau^t \left(\frac{\sigma}{2} -
         \frac{2}{\nu}C(\Om)^4\|w(s)\|^2\right)ds\}.
    \end{equation}
From (\ref{3.19}) it follows that the solution $w$ of Problem \ref{pr.2.1} belongs to $L^2(\tau, t; V)$.
By (\ref{eqn:er2.15a}) the map $v(\tau)\to v(t)$, $t>\tau \geq 0$, in $H$ is Lipschitz continuous, with
\begin{eqnarray} \label{ineq:333}
     |w(t) - v(t)| \leq C |w(\tau) - v(\tau)|
\end{eqnarray}
uniformly for $t, \tau$ in a given interval $[0,T]$ and initial conditions $w(0), v(0)$ in a given 
bounded set $B$ in $H$.

In particular, as $u(0)= w(0)-v(0)=0$,  the solution  $v$   of Problem \ref{pr.2.1} is unique.
This ends the proof of Theorem \ref{thm3.2}.
\end{proof}

\renewcommand{\theequation}{\arabic{section}.\arabic{equation}}
\setcounter{equation}{0}
\section{Preliminaries from the theory of dynamical systems} \label{l-trajectories}
Let us consider an abstract autonomous evolutionary problem
\begin{eqnarray} \label{tra:1}
    \frac{dv(t)}{dt} &=& F(v(t))   \quad {\rm in} \quad X,  \\
    v(0)&=& v_0. \nonumber
\end{eqnarray}
where $X$ is a Banach space, $F: X\to X$ is a nonlinear operator, and $v_0\in X$. We assume that the above problem has a global in time unique solution $[0,\infty ) \ni t \to v(t) \in X$ for every $v_0\in X$. In this case one can associate with the problem a semigroup $\{S(t)\}_{t\geq 0}$ of (nonlinear) operators $S(t):X\to X$ setting 
$S(t)v_0=v(t)$, where $v(t)$, $t>0$, is the unique solution of (\ref{tra:1}). 

From properties of the semigroup of operators $\{S(t)\}_{t\geq 0}$ we may then conclude the basic features of the behaviour of solutions of problem (\ref{tra:1}), in particular, their time asymptotics. One of the objects existence of which characterize the asymptotic behaviour of solutions is the global attractor. It is a compact and invariant with respect to operators $S(t)$ subset of the phase space $X$ (in general, a metric space) that uniformly attracts all bounded subsets of $X$. 
\begin{definition} A global attractor for a semigroup $\{S(t)\}_{t\geq 0}$ in a Banach space $X$ is 
a subset ${\cal A}$ of $X$ such that 
\begin{itemize}
\item ${\cal A}$ is compact in X.

\item ${\cal A}$ is invariant, i.e., $S(t){\cal A}={\cal A}$ for every $t\geq 0$.

\item For every $\varepsilon >0$ and every bounded set $B$ in $X$ there exists $t_0=t_0(B,\varepsilon)$ such that for all $t\geq t_0$, $S(t)B$ is a subset of the $\varepsilon$-neighbourhood of the attractor ${\cal A}$ (uniform attraction property).
\end{itemize}
\end{definition}
The global attractor defined above is uniquely determined by the semigroup $\{S(t)\}_{t\geq 0}$. Morever, it is connected and also has the following properties: it is the maximal compact invariant set and the minimal set that attracts all bounded sets. The global attractor may have a very complex structure. However, as a compact set (in an infinite dimensional Banach space) its interior is empty. For many dynamical systems the global attractor has a finite fractal dimension (defined below) which has a number of important consequences for the behaviour of the flow generated by the semigroup
\cite{robinson-2011-dim, robinson-2001-infty}.
\begin{definition} The fractal dimension of a compact set $K$ in a Banach space $X$ is defined as
\begin{eqnarray*} \label{tra:2}
    d_f^X(K) = \lim\sup_{\varepsilon\to 0}\frac{\log N_\varepsilon^X(K)}{\log(\frac{1}{\varepsilon})}
\end{eqnarray*}
where $N_\varepsilon^X(K)$ is the minimal number of balls of radius $\varepsilon$ in $X$ needed to cover $K$.
\end{definition}
Another important property that holds for many dynamical systems is the existence of an exponential attractor. 
\begin{definition} An exponential attractor for a semigroup $\{S(t)\}_{t\geq 0}$ in a Banach space $X$ is 
a subset ${\cal M}$ of $X$ such that 
\begin{itemize}
\item ${\cal M}$ is compact in X.

\item ${\cal M}$ is positively invariant, i.e., $S(t){\cal M} \subset {\cal M}$ for every $t\geq 0$.

\item Fractal dimension of ${\cal M}$ is finite, i.e., $d_f^X({\cal M}) < \infty$.

\item ${\cal M}$ attracts exponentially the images of bounded subsets of $X$, i.e., there exist a universal constant $c_1$ and a monotone function $\Phi$ such that for every bounded set $B$ in $X$, its image $S(t)B$ is a subset 
of the $\varepsilon(t)$-neighbourhood of ${\cal M}$ for all $t\geq t_0$, 
where $\varepsilon(t)= \Phi(||B||_X) e^{-c_1t}$ (exponential attraction property).
\end{itemize}
\end{definition}
In the following sections we shall consider the problem of the existence of the global and an exponential attractor for the dynamical system considered in this paper. 

Let $X$, $Y$, and $Z$ be three Banach spaces such that 
\begin{eqnarray*} \label{tra:3}
    Y\subset X \quad {\rm with \,\,compact\,\,imbedding} \quad {\rm and} \quad X\subset Z.
\end{eqnarray*}
We assume, moreover, that $X$ is reflexive and separable. 

For $\tau >0$, let
\begin{eqnarray*} \label{tra:4}
    X_\tau = L^2(0,\tau;X),
\end{eqnarray*}
and
\begin{eqnarray*} \label{tra:5}
    Y_\tau = \{u\in L^{p_1}(0,\tau;Y),\,\,\frac{du}{dt}\in L^{p_2}(0,\tau;Z)\},
\end{eqnarray*}
for some $2\leq p_1 < \infty$ and $1 \leq p_2 < \infty$.

By $C([0,\tau];X_w)$ we denote the space of weakly continuous functions from the interval $[0,\tau]$ to the Banach space $X$, and we assume that the solutions of (\ref{tra:1}) are at least in $C([0,T];X_w)$ for all $T>0$. Then by an 
$l$-trajectory we mean parts of solution trajectories parametrized by time from the interval $[0,l]$. If 
$v=v(t), t>0$, is the solution of (\ref{tra:1}) then $\chi=v|_{[0,l]}$ is an $l$-trajectory as well as all shifts
$L_t(v)(\tau)=v(t+\tau), 0\leq \tau\leq l$, for $t>0$. 

We can now formulate a theorem which gives criteria for the existence of a~global attractor ${\cal A}$ for the semigroup $\{S(t)\}_{t\geq 0}$ in $X$ and its finite dimensionality. These criteria are stated as assumptions (A1)-(A8) in 
\cite{malek-prazak-2002}.
\begin{itemize}
\item[(A1)] For any $v_0\in X$ and arbitrary $T>0$ there exists (not necessarily unique)
    $v\in C([0,T];X_w) \cap Y_T$, a solution of the evolutionary problem on $[0,T]$ with $v(0)=v_0$. Moreover,
    for any solution the estimates of $||v||_{Y_T}$ are uniform with respect to $||v(0)||_X$.
    
\item[(A2)] There exists a bounded set $B^0\subset X$ with the following properties: if $v$ is an arbitrary
solution with initial condition $v_0\in X$ then (i) there exists $t_0=t_0(||v_0||_X)$ such that $v(t)\in B^0$ for 
all $t\geq t_0$ and (ii) if $v_0\in B^0$ then $v(t)\in B^0$ for all $t\geq 0$.

\item[(A3)] Each $l$-trajectory has among all solutions a unique continuation which means that from an end point of an $l$-trajectory there starts at most one solution.

\item[(A4)] For all $t>0$, $L_t: X_l \to X_l$ is continuous on ${\cal B}_0^l$ - the set of all $l$-trajectories starting at any point of $B^0$ from (A2).

\item[(A5)] For some $\tau > 0$, the closure in $X_l$ of the set $L_\tau ({\cal B}_0^l)$ is included in ${\cal B}_0^l$.

\item[(A6)] There exists a space $W_l$ such that $W_l\subset X_l$ with compact embedding, and $\tau >0$ such that $L_\tau:X_l\to W_l$ is Lipschitz continuous on ${\cal B}_l^1$ - the closure of $L_\tau({\cal B}_0^l)$ in $X_l$.

\item[(A7)] The map $e:X_l \to X$, $e(\chi)=\chi(l)$ is continuous on ${\cal B}^1_l$.

\item[(A8)] The map $e:X_l \to X$ is H\"older-continuous on ${\cal B}^1_l$.
\end{itemize}
\begin{theorem} \label{malek1}
Let the assumptions (A1)-(A5), (A7) hold. Then there exists a global attractor ${\cal A}$ for the semigroup
$\{S(t)\}_{t\geq 0}$ in $X$. Moreover, if the assumptions (A6), (A8) are satisfied then the fractal dimension of the attractor is finite.
\end{theorem}
For the existence of an exponential attractor we need two additional properties to hold, where now $X$ is a Hilbert space 
(cf., \cite{malek-prazak-2002}).
\begin{itemize}
\item[(A9)] For all $\tau>0$ the operators $L_t:X_l\to X_l$ are (uniformly with respect to $t\in[0,\tau]$) Lipschitz continuous on ${\cal B}_l^1$.

\item[(A10)] For all $\tau>0$ there exists $c>0$ and $\beta \in(0,1]$ such that for all $\chi\in {\cal B}_l^1$ and
$t_1, t_2 \in[0,\tau]$ it holds that 
\begin{eqnarray} \label{exponential-problem}
      ||L_{t_1}\chi - L_{t_2}\chi||_{X_l} \leq c |t_1 -t_2|^\beta.
\end{eqnarray}
\end{itemize}
\begin{theorem} \label{existence-exponential} 
Let $X$ be a separable Hilbert space and let the assumptions (A1)-(A6) and (A8)-(A10) hold. Then there exists an exponential attractor ${\cal M}$ for the semigroup $\{S(t)\}_{t\geq 0}$ in $X$. 
\end{theorem}
For the proofs of Theorems \ref{malek1} and \ref{existence-exponential} we refer the readers to corresponding theorems in 
\cite{malek-prazak-2002}.

\renewcommand{\theequation}{\arabic{section}.\arabic{equation}}
 \setcounter{equation}{0}
\section{Existence of the global attractor of a finite fractal dimension} \label{global-attractor}
In this section we prove the following theorem.
\begin{theorem} \label{theorem-main}
There exists a global attractor of a finite fractal dimension for the semigroup $\{S(t)\}_{t\geq 0}$ associated with Problem \ref{pr.2.1}.
\end{theorem}
\begin{proof} From the considerations in the previous section it follows that to prove the theorem it suffices to check assumptions (A1)-(A6), (A8). For the convenience of the reader we repeat their statements in appropriate places.

\vspace{0.3cm}

{\it Assumption} (A1). For any $v_0\in X$ and arbitrary $T>0$ there exists (not necessarily unique)
    $v\in C([0,T];X_w) \cap Y_T$, a solution of the evolutionary problem on $[0,T]$ with $v(0)=v_0$. Moreover,
    for any solution the estimates of $||v||_{Y_T}$ are uniform with respect to $||v(0)||_X$.
    
\vspace{0.3cm} In our case, set $X=H$, $Y=V$, and $Y_T = \{u\in L^2(0,T;V), u' \in L^2(0,T; V')\}$. From Theorems \ref{thm3.1}
and \ref{thm3.2} we know that for any $v_0\in H$ and arbitrary $T>0$ there exists a unique 
    $v\in C([0,T];H) \cap Y_T$,  solution of Problem \ref{pr.2.1}.
We shall obtain the needed estimates directly from the variational inequality, cf. (\ref{eqn:er2.16}),    
\begin{eqnarray} \label{eqn:va4.1}
     \langle v_{t}(t),  \Theta-v(t) \rangle\, + \,\,\nu a(v(t), \Theta- v(t)) \,
            &+& \, b(v(t), v(t), \Theta - v(t))  
           \\
      &+& j(\Theta) - j( v(t))\, \geq \, ({\cal L}(v(t)), \Theta- v(t)) \nonumber
\end{eqnarray}    
Set $\Theta=0$ in (\ref{eqn:va4.1}) to get
\begin{eqnarray*} \label{eqn:va4.2}
     \frac{1}{2}\frac{d}{dt}|v|^2 + \nu||v||^2 + j(v) \leq ({\cal L}(v), v),
\end{eqnarray*}  
as $j(0)=0$. Since, by Lemma \ref{lemma3.1},
\begin{eqnarray*} \label{eqn:va4.3}
     ({\cal L}(v), v) \leq \frac{\nu}{2}||v||^2 + \nu ||\xi||^2,
\end{eqnarray*}  
we obtain
\begin{eqnarray} \label{eqn:va4.4}
     \frac{d}{dt}|v|^2 + \nu||v||^2 + 2j(v) \leq 2\nu ||\xi||^2 = F.
\end{eqnarray}
 Integrating in $t$ we obtain
\begin{equation} \label{eqn:va4.5}
     |{v(t)}|^{2} + \nu\int_{0}^{t}\|{v(s)}\|^{2} ds \leq |v(0)|^{2} + 2tF
\end{equation}
and we deduce that
 \begin{eqnarray*}\label{eqn:va4.6}
      v \mbox{ is bounded in } L^{2}(0 , T ; V)\cap L^{\infty}(0, T ; H), 
 \end{eqnarray*}
uniformly with respect to $|v(0)|$.

\vspace{0.3cm}

To get a uniform with respect to $|v(0)|$ estimate of $v'$ in $L^2(0,T;V')$ set $\Theta=v-\psi$, $\psi\in V$, 
in (\ref{eqn:va4.1}). We then have
\begin{equation}\label{eqn:va4.7}
      \langle v', \psi \rangle \,\,\leq \, ({\cal L}(v), \psi) - \nu a(v,\psi) - b(v, v, \psi) + j(v-\psi) -j(v).
\end{equation}
Thanks to the Poincar\'e inequality we have, $||\gamma(v)||_{L^2(\partial\Omega)} \leq C ||v||$, and
\begin{equation}\label{eqn:va4.8}
      j(v-\psi) -j(v) = k\int_{\Gamma_0} (|v-\psi|-|v|) \leq  k\int_{\Gamma_0}|\psi| \leq C(\Gamma_0)||\psi||. 
\end{equation}
Moreover, using the Ladyzhenskaya inequality (\ref{lady}) to the nonlinear term, we have
\begin{equation}\label{eqn:va4.9}
     ({\cal L}(v), \psi) - \nu a(v,\psi) - b(v, v, \psi) \leq C_1(||v|| + |v|\,||v|| + 1)||\psi||.
\end{equation}
From (\ref{eqn:va4.7})-(\ref{eqn:va4.9}) we obtain
\begin{equation}\label{eqn:va4.10}
     ||v'||_{V'} \leq C_2(||v|| + |v|\,||v|| + 1) \nonumber
\end{equation}
and
\begin{eqnarray*}\label{eqn:va4.11}
     ||v'||_{L^2(0,T;V')}^2 &=& \int_0^T ||v'(t)||_{V'}^2dt  \\  &\leq& C_2\left(\int_0^T||v(t)||^2dt 
                           + ||v||_{L^{\infty}(0,T;H)}^2\int_0^T||v(t)||^2dt + T\right)  
                           \leq C(|v(0)|). \nonumber
\end{eqnarray*}
Thus, (A1) holds true.

\vspace{0.3cm}

{\it Assumption} (A2). There exists a bounded set $B^0\subset X$ with the following properties: if $v$ is an arbitrary
solution with initial condition $v_0\in X$ then (i) there exists $t_0=t_0(||v_0||_X)$ such that $v(t)\in B^0$ for 
all $t\geq t_0$ and (ii) if $v_0\in B^0$ then $v(t)\in B^0$ for all $t\geq 0$.

\vspace{0.3cm}

From (\ref{eqn:va4.4}) and the Poincar\'e inequality, 
\begin{eqnarray*} \label{eqn:va4.12}
     \frac{d}{dt}|v|^2 + \nu\lambda_1 |v|^2 \leq F,
\end{eqnarray*}
and then by the Gronwall lemma,
\begin{eqnarray*} \label{eqn:va4.13}
     |v(t)|^2 \leq |v(0)|^2 e^{-\nu\lambda_1 t} + \frac{F}{\nu\lambda_1}.
\end{eqnarray*}
Thus, there exists a bounded absorbing set (e.g., the ball $B_H(0, \rho)$ with $\rho^2=2\frac{F}{\nu\lambda_1}$) in $H$.
Let $t_0$ be a time at which $B_H(0, \rho)$ absorbs itself and let $B^0$ be the closure of $S(t_0)B_H(0, \rho)$ in $H$. 
If $v_0\in B^0$ then $v(t)\in B^0$ for all $t\geq 0$. Thus, (A2) holds true. (We need $B^0$ to be closed in $H$ to be able to satisfy assumption (A5) below).

\vspace{0.3cm}

{\it Assumption} (A3). Each $l$-trajectory has among all solutions a unique continuation.

\vspace{0.3cm}

We recall that by the $l$-trajectory we mean any solution on the time interval $[0, l]$, and the unique continuation 
means that from an end point of an $l$-trajectory there starts at most one solution. 

In our case (A3) is satisfied as the solutions are unique (Theorem \ref{thm3.2}). 

\vspace{0.3cm}

{\it Assumption} (A4). For all $t>0$, $L_t: X_l \to X_l$ is continuous on ${\cal B}_0^l$.

\vspace{0.3cm}

${\cal B}_0^l$ is defined as the set of all $l$-trajectories starting at any point of $B^0$ from (A2), and 
$X_l=L^2(0,l;X)$. The semigroup $\{L_t: t\geq 0\}$ acts on the set of $l$-trajectories as the shifts operators: $\{L_t\chi\}(\tau) = v(t+\tau)$ 
for $0\leq \tau \leq l$, where $v$ is the unique solution on $[0, l+\tau]$ such that $v|_{[0,l]}=\chi$.

\vspace{0.3cm}

In our case, $X = H$, hence $X_l=H_l=L^2(0,l;H)$. In view of inequality (\ref{ineq:333}) the map $S(t):B^0 \to B^0$ is Lipschitz continuous for every $t>0$, and we have, for any two $\chi_1, \chi_2$ in ${\cal B}_0^l$,
\begin{eqnarray} \label{eqn:va4.14}
     \int_0^l|L_t\chi_1(s) - L_t\chi_2(s)|^2ds \leq C^2(t) \int_0^l|\chi_1(s) - \chi_2(s)|^2ds.
\end{eqnarray}
Thus, (A4) holds true.

\vspace{0.3cm}

{\it Assumption} (A5). For some $\tau > 0$, the closure in $X_l$ of the set $L_\tau ({\cal B}_0^l)$ is included in ${\cal B}_0^l$.

\vspace{0.3cm}

As $L_\tau ({\cal B}_0^l)\subset {\cal B}_0^l$, it is enough to check that the set ${\cal B}_0^l$ is closed in $X_l$. We have to prove that if $\{\chi_n\}$ is a sequence in ${\cal B}_0^l$ converging to some $\chi$ in $X_l$ then $\chi$ is also a trajectory and that 
  $\chi(0)\in B^0$.

From Assumption (A1) it follows that the sequence $\{\chi_n\}$ is bounded in $Y_l$ and thus contains a subsequence (relabeled $\{\chi_n\}$) such that
\begin{eqnarray}\label{eqn:va4.15}
    \chi_n\tow \chi \quad {\rm in} \quad L^{2}(0 , l ; V),\quad {\rm and}  \quad
    \dfrac{d \chi_n}{dt}\tow  \dfrac{d \chi}{dt}\quad {\rm in} \quad L^{2}(0 , l ; V').
\end{eqnarray}
Moreover, by the Aubin-Lions lemma,
\begin{eqnarray}\label{eqn:va4.15a}
    \chi_n\to \chi \quad {\rm in} \quad L^{2}(0 , l ; H).
\end{eqnarray}
We have,
\begin{eqnarray*} \label{eqn:va4.16}
     \langle \chi_n'(t), \Theta-\chi_n(t) \rangle\, &+& \,\nu a(\chi_n(t), \Theta- \chi_n(t)) \, \nonumber\\ \nonumber \\
            &+& \, b(\chi_n(t), \chi_n(t), \Theta - \chi_n(t)) 
           \nonumber\\ \nonumber \\
      &+& j(\Theta) - j(\chi_n(t)) \geq  ({\cal L}(\chi_n(t)) \,,\, \Theta- \chi_n(t)).
\end{eqnarray*}
We multiply both sides by a nonnegative smooth function $\eta=\eta(t)$ with support in the interval $(0, l)$ and integrate with respect to $t$ in this interval. We shall prove that taking $\lim \inf_{n\to\infty}$ of both sides and using (\ref{eqn:va4.15}) and  (\ref{eqn:va4.15a}), we obtain
\begin{eqnarray} \label{eqn:va4.17}
     \int_0^l\langle\chi'(t), \Theta -\chi(t)\rangle\eta(t)dt &+&\nu\int_0^l a(\chi(t) , \Theta- \chi(t))\eta(t)dt \nonumber\\ \nonumber \\
            &+& \int_0^l b(\chi(t) , \chi(t), \Theta - \chi(t))\eta(t)dt  
           \nonumber\\ \nonumber \\
      &+&\int_0^l j(\Theta)\eta(t)dt - \int_0^lj(\chi(t))\eta(t)dt   \nonumber\\ \nonumber \\
      &\geq&\int_0^l({\cal L}(\chi(t)) ,\Theta- \chi(t))\eta(t)dt.
\end{eqnarray}
First we shall show (without using the convexity argument) that
\begin{eqnarray} \label{eqn:va4.18}
    \lim_{n\to\infty} \int_0^lj(\chi_n(t))\eta(t)dt  =  \int_0^lj(\chi(t))\eta(t)dt.
\end{eqnarray}
It is known (cf., e.g., \cite{necas-1967}) that for every $\varepsilon >0$ there exists $C_\varepsilon >0$ such that for all $v\in W^{1,2}(\Omega)$,
\begin{eqnarray*} \label{eqn:va4.19}
    ||\gamma(v)||_{L^2(\partial\Omega)} \leq \varepsilon||\nabla v||_{L^2(\Omega)} + C_\varepsilon||v||_{L^2(\Omega)}.
\end{eqnarray*}
\noindent We have thus
\begin{eqnarray*} \label{eqn:va4.20}
    ||\gamma(\chi_n)-\gamma(\chi)||_{L^2(\Gamma_0)}^2 \leq \varepsilon|| \chi_n - \chi||^2 
    + C'_\varepsilon|\chi_n-\chi|^2
\end{eqnarray*}
and 
\begin{eqnarray*} \label{eqn:va4.21}
    \int_0^l ||\gamma(\chi_n)-\gamma(\chi)||_{L^2(\Gamma_0)}^2dt \leq 
    \varepsilon\int_0^l|| \chi_n - \chi||^2dt 
    + C'_\varepsilon\int_0^l|\chi_n-\chi|^2dt.
\end{eqnarray*}
As, in view of (\ref{eqn:va4.15}), there exists $M>0$ such that for all $n$,
\begin{eqnarray*} \label{eqn:va4.22}
    \int_0^l||\chi_n(t) - \chi(t)||^2dt \leq M,
\end{eqnarray*}
we obtain, using (\ref{eqn:va4.15a}),
\begin{eqnarray*} \label{eqn:va4.23}
    \lim\sup_{n\to\infty}\int_0^l ||\gamma(\chi_n)-\gamma(\chi)||_{L^2(\Gamma_0)}^2dt \leq \varepsilon M. 
\end{eqnarray*}
Now, as $\varepsilon$ is any positive number, we obtain
\begin{eqnarray} \label{eqn:va4.24}
    \lim_{n\to\infty}\int_0^l ||\gamma(\chi_n)-\gamma(\chi)||_{L^2(\Gamma_0)}^2dt = 0. 
\end{eqnarray}
From (\ref{eqn:va4.24}), (\ref{eqn:va4.18}) easily follows.

We have also
\begin{eqnarray*} \label{eqn:va4.25}
    \lim_{n\to\infty}\int_0^l \langle\chi_n',\chi_n\rangle\eta(t)dt &=& 
    \lim_{n\to\infty}\int_0^l \frac{1}{2}\frac{d}{dt}|\chi_n|^2\eta(t)dt \nonumber \\
    &=&-\lim_{n\to\infty}\int_0^l \frac{1}{2}|\chi_n|^2\eta(t)'dt = -\int_0^l \frac{1}{2}|\chi|^2\eta(t)'dt \nonumber \\
    &=& \int_0^l \langle\chi',\chi\rangle\eta(t)dt.
\end{eqnarray*}
In view of (\ref{eqn:va4.15}) and  (\ref{eqn:va4.15a}), there are no problems to get the other terms in 
(\ref{eqn:va4.17}) and finally, (\ref{eqn:va4.17}) itself. As inequality (\ref{eqn:va4.17}) is equivalent to inequality
\begin{eqnarray*} \label{eqn:va4.26}
     \langle\chi(t)', \Theta-\chi \rangle\, + \, \nu a(\chi(t), \Theta- \chi(t)) 
            &+& b(\chi(t) ,\chi(t), \Theta - \chi(t))  \\
           &+& j(\Theta) - j(\chi(t)) \geq  ({\cal L}(\chi(t)), \Theta- \chi(t)) \nonumber\\ \nonumber 
\end{eqnarray*}
satisfied for almost all $t\in (0,l)$, $\chi$ is a solution with $\chi(0)$ in $H$. To end the proof we have to 
show that $\chi(0)$ belongs to the positively absorbing set $B^0$.  We have $\chi_n(t)\in B^0$ for all $t\in [0,l]$ and, by 
(\ref{eqn:va4.15a}), for a subsequence, $\chi_n(t) \to \chi(t)$ for almost all $t\in (0,l)$. As $B^0$ is closed, $\chi(t)\in B^0$
for almost all $t\in (0,l)$. Now, from the continuity of $\chi:[0,l] \to H$ and the closedness of $B^0$ it follows that
$\chi(0)$ is in $B^0$.

Thus, assumption (A5) holds.

\vspace{0.3cm}

{\it Assumption} (A6).  There exists a space $W_l$ such that $W_l\subset X_l$ with compact embedding, and $\tau >0$ such that $L_\tau:X_l\to W_l$ is Lipschitz continuous on ${\cal B}_l^1$ - the closure in $X_l$ 
of $L_\tau({\cal B}_l^0)$.

\vspace{0.3cm}

Define,
\begin{eqnarray*} \label{eqn:va4.27}
     W_l = \{u: u\in L^2(0,l;V), u'\in L^1(0,l;U')\}.
\end{eqnarray*}
where $U=\{\psi\in V: \psi = 0 \,\,\, at \,\,\,\Gamma_0 \}$.
We have, $W_l\subset X_l$, with compact embedding and we shall prove that $L_l:X_l\to W_l$ is Lipschitz continuous on ${\cal B}_l^1$.

Let $w$ and $v$ be two solutions of Problem~\ref{pr.2.1} starting from $B^0$ and let $u=w-v$. Then, cf. (\ref{eqn:er2.14e}),
\begin{equation} \label{eqn:er2.14e-bis}
           \frac{d}{dt}|{u(t)}|^{2} + \frac{\nu}{2}
            \|{u(t)}\|^{2} \leq \frac{2}{\nu}C(\Om)^4\|w(t)\|^2 |u(t)|^2. \nonumber
\end{equation}
Take $s\in (0,l)$ and integrate this inequality over $\tau\in (s,2l)$ to get
\begin{equation} \label{eqn:va4.28}
           |{u(2l)}|^{2} + \frac{\nu}{2}\int_s^{2l} \|{u(\tau)}\|^{2}d\tau 
                        \leq \frac{2}{\nu}C(\Om)^4\int_s^{2l}\|u(\tau)\|^2 |u(\tau)|^2d\tau
                         + |u(s)|^2.
\end{equation}
From (\ref{ineq:333}) we conclude that
\begin{eqnarray*} \label{ineq:333-bis}
     |u(\tau)|^2 \leq C_1 |u(s)|^2
\end{eqnarray*}
for $\tau\in (s,2l)$ and from (\ref{eqn:va4.5}) we have
\begin{eqnarray*}
     \int_s^{2l} \|{u(\tau)}\|^{2}d\tau \leq \int_0^{2l}(||v(\tau)||^2 + ||w(\tau)||^2)d\tau \leq \frac{1}{\nu}(|v(0)|^2+|w(0)|^2+8lF),
\end{eqnarray*}
whence
\begin{equation} \label{eqn:va4.29}
           \int_s^{2l} \|{u(\tau)}\|^{2}d\tau 
                        \leq C_2   \nonumber
\end{equation}
uniformly for  $w(s),v(s)\in B^0$. Therefore, from (\ref{eqn:va4.28}) we obtain
\begin{equation} \label{eqn:va4.30}
           \int_l^{2l} \|{u(\tau)}\|^{2}d\tau \leq C_3|u(s)|^2.  \nonumber
\end{equation}
Integrating over $s\in (0,l)$ we obtain
\begin{equation} \label{eqn:va4.31}
           \int_l^{2l} \|{u(\tau)}\|^{2}d\tau \leq \frac{C_3}{l}\int_0^l|u(s)|^2 ds \nonumber
\end{equation}
and therefore
\begin{equation} \label{eqn:va4.32}
   ||L_l\chi_1 -L_l\chi_2||_{L^2(0,l;V)} \leq \sqrt{\frac{C_3}{l}} ||\chi_1 - \chi_2||_{L^2(0,l;H)}
\end{equation}
for any $\chi_1, \chi_2$ in ${\cal B}_l^0$.

In order to prove that
\begin{equation} \label{eqn:va4.33}
           ||(L_l\chi_1 -L_l\chi_2)'||_{L^1(0,l;U')} \leq C ||\chi_1 - \chi_2||_{L^2(0,l;H)} \nonumber
\end{equation}
for some $C>0$ it is sufficient, in view of (\ref{eqn:va4.32}), to prove that 
\begin{equation} \label{eqn:va4.34}
           ||(\chi_1 - \chi_2)'||_{L^1(0,l;U')} \leq C' ||\chi_1 - \chi_2||_{L^2(0,l;V)}  
\end{equation}
with some $C'>0$.
We have,
\begin{eqnarray*} \label{eqn:er2.16-bis}
     \langle v' ,  \Theta-v\rangle\, + \,\nu a(v, \Theta- v) \,
            + \, b(v , v , \Theta - v)  \,
      + j(\Theta) - j( v) \geq  ({\cal L}(v) , \Theta- v)
\end{eqnarray*} 
and
\begin{eqnarray*} \label{eqn:er2.16-bis-bis}
     \langle w' ,  \Theta-w \rangle\, + \,\nu a(w, \Theta- w) \,
            + \, b(w , w , \Theta - w)  \,
      + j(\Theta) - j( w) \geq  ({\cal L}(w) , \Theta- w).
\end{eqnarray*}

Set $\Theta= v-\psi$ in the first inequality and $\Theta=w+\psi$ in the second one, 
where $\psi\in U$, $||\psi||\leq 1$, and add thus obtained inequalities to get
\begin{eqnarray} \label{eqn:va4.35}
     \langle u', -\psi \rangle \, \leq {\cal N}(u,w,v;\psi),
\end{eqnarray}

where
\begin{eqnarray*} \label{eqn:va4.36}
     {\cal N}(u,w,v;\psi) = b(\xi,u,\psi) + b(u,\xi,\psi) + \nu a(u,\psi) + b(w,u,\psi) + b(u,v,\psi).
\end{eqnarray*}
Estimating the right hand side of (\ref{eqn:va4.35}) we get
\begin{eqnarray*} \label{eqn:va4.38}
     \langle u', -\psi \rangle\,\, \leq \,\,C''(1+||v| + ||w||)||u||\,||\psi||,
\end{eqnarray*}
whence
\begin{eqnarray*} \label{eqn:va4.39}
    ||u'(t)||_{U'} = \sup\{ \langle u'(t), -\psi \rangle\,: ||\psi||\leq 1\} \,\, \leq \,\,C''(1+ ||v|| + ||w(t)||)||u(t)||.
\end{eqnarray*}
At last, integration over $t\in (0,l)$ gives
\begin{eqnarray*} \label{eqn:va4.40}
    \int_0^l||u'(t)||_{U'}dt  \leq \,\,C'''(\int_0^l||u(t)||^2dt)^{1/2},
\end{eqnarray*}
with
\begin{eqnarray*} \label{eqn:va4.41}
    C'''= C_0(\int_0^l (1+||v(t)||^2 + ||w(t)||^2)dt)^{1/2}, \quad C_0>0,
\end{eqnarray*}
uniformly for trajectories starting from $B^0$. This proves (\ref{eqn:va4.34}) and ends the proof of the Lipschitz continuity of the map $L_l:X_l\to W_l$.
 Assumption (A6) holds true.

\vspace{0.3cm}



\vspace{0.3cm}

{\it Assumption} (A8). The map $e:X_l \to X$ is H\"older-continuous on ${\cal B}^1_l$.

\vspace{0.3cm}

\noindent (A8) follows directly from the Lipschitz continuity of the map $e:X_l \to X$, $e(\chi)=\chi(l)$. 
To check the latter, 
let $w,v$ be two solutions as above, starting from $B^0$. From (\ref{ineq:333}) we have, in particular,
\begin{eqnarray*} \label{eqn:va4.42}
     |w(l) - v(l)| \leq C |w(\tau) - v(\tau)|
\end{eqnarray*}
for $\tau\in (0,l)$. Integrating this inequality in $\tau$ on the interval $(0,l)$ we obtain
\begin{eqnarray*} 
     |w(l) - v(l)| \leq \frac{C}{\sqrt{l}} ||w - v||_{L^2(0,l;H)}.
\end{eqnarray*}

\vspace{0.3cm}
\noindent This ends the proof of Theorem \ref{theorem-main}. 
\end{proof}

\renewcommand{\theequation}{\arabic{section}.\arabic{equation}}
\setcounter{equation}{0}
\section{Existence of an exponential attractor} \label{exponential}
In this section we prove the main theorem of this paper. 

\begin{theorem} \label{thm-exponential}
There exists an exponential attractor for the semigroup $\{S(t)\}_{t\geq 0}$ associated with Problem \ref{pr.2.1}.
\end{theorem}
\begin{proof}
In view of Theorem \ref{existence-exponential} and the considerations in the previous section it suffices to check conditions (A9) and (A10). The first one follows immediately from inequality (\ref{eqn:va4.14}). Thus it remains to prove that 
\begin{eqnarray} \label{exponential-problem2}
      ||L_{t_1}\chi - L_{t_2}\chi||_{X_l} \leq c |t_1 -t_2|^\beta
\end{eqnarray}
holds for all $\tau>0$, $0\leq t_2\leq t_1 \leq \tau$, $\chi\in {\cal B}_l^1$, some $c>0$ and some $\beta \in(0,1]$, where in our case, $X_l=H_l$.  

To obtain (\ref{exponential-problem2}) it suffices to know that $\chi'$, the time derivatives 
of $\chi\in {\cal B}_l^1$, are uniformly bounded in $L^q(0,l; H)$ for some $1< q \leq \infty$, 
\cite{malek-prazak-2002}. In fact, we have then
\begin{eqnarray*} \label{eqn:exponential1}
      ||L_{t_1}\chi(s) - L_{t_2}\chi(s)||_{H} &=& ||u(t_1+s) - u(t_2+s)||_H \\
      &=& ||\int_{t_2+s}^{t_1+s}u'(\eta)d\eta || \\
      &\leq & |t_1 -t_2|^{1-\frac{1}{q}}||u'||_{L^q(t_2+s,t_1+s;H)}.
\end{eqnarray*}
and integration with respect to $s$ in the interval $[0,l]$ gives (\ref{exponential-problem2}) with $c$ depending on $\tau$ and $l$.

We shall prove that there exists $M>0$ such that
\begin{eqnarray*} \label{eqn:exponential2}
      ||\chi'||_{L^\infty(0,l;H)} \leq M \quad \mbox{for all} \quad \chi\in {\cal B}_l^1.
\end{eqnarray*}
The formal a priori estimates that follow can be performed on the smooth in the time variable Galerkin approximations $v^n_\delta(t)$, $n=1,2,3,...$, of the regularized problem (\ref{e32})-(\ref{e320}), as in 
\cite{Lady-Seregin-1995}. The obtained estimates are preserved by solutions $v_\delta(t)$ of 
problem (\ref{e32})-(\ref{e320}), with bounds independent of $\delta$ when $\delta\to 0$, and in the end, by solutions
$v(t)$ of Problem \ref{pr.2.1}.

Let us consider the solution $v=v(t)$ of problem (\ref{e32})-(\ref{e320}) (we drop the subscript $\delta$ for short),
\begin{eqnarray}\label{eqn;exp3}
     (\frac{d v(t)}{dt} ,\Theta) + \nu a(v(t), \Theta) &+&
      b(v(t), v(t),\Theta)
      + (j'_{\delta}(v(t)), \Theta ) \nonumber\\
     &=&
     - \nu a(\xi, \Theta)
     -  b(\xi , v(t) , \Theta)- b(v(t) ,  \xi , \Theta)
\end{eqnarray}
\noindent with initial condition
\begin{eqnarray*}\label{eqn:exp4}
v(0)= v_{0}.
\end{eqnarray*}
Our aim is to derive, following the method used in \cite{Lady-Seregin-1995}, two a priori estimates which yield (\ref{eqn:exponential2}). To get the first one, set 
$\Theta = v_t$ ($v_t=v'$) in (\ref{eqn;exp3}). We obtain
\begin{eqnarray}\label{eqn:exp5}
    |v_t|^2 + \frac{d}{dt}\frac{\nu}{2}||v||^2 + \frac{d}{dt}j_\delta(v) = ({\cal L}(v), v_t) - b(v,v,v_t).
\end{eqnarray}
Using the Ladyzhenskaya inequality (\ref{lady}) we have
\begin{eqnarray}\label{eqn:exp6}
    ({\cal L}(v), v_t) \leq c_1(||v_t|| + ||v|||v_t|^{1/2}||v_t||^{1/2} + |v|^{1/2}||v||^{1/2}|v_t|^{1/2}||v_t||^{1/2}).
\end{eqnarray}
Now, we differentiate (\ref{eqn;exp3}) with respect to the time variable and set $\Theta = v_t$, to get
\begin{eqnarray}\label{eqn:exp7}
    \frac{1}{2}\frac{d}{dt}|v_t|^2 + \nu||v_t||^2 \leq -b(v_t, \xi, v_t) - b(v_t, v, v_t),
\end{eqnarray}
as $b(\xi, v_t, v_t)=0$, $b(v, v_t, v_t)=0$, and, cf., \cite{Duv72}, \cite{papa-1985},
\begin{eqnarray*}\label{eqn:exp8}
    (\frac{d}{dt}j'_\delta(v), v_t) \geq 0.
\end{eqnarray*}
Using Lemma \ref{lemma3.1} and the Ladyzhenskaya inequality (\ref{lady}) to estimate the right hand side of (\ref{eqn:exp7}) we obtain
\begin{eqnarray}\label{eqn:exp9}
    \frac{d}{dt}|v_t|^2 + \nu||v_t||^2 \leq c_2||v||^2|v_t|^2.
\end{eqnarray}
Now, we multiply (\ref{eqn:exp9}) by $t^2$ to get
\begin{eqnarray}\label{eqn:exp10}
    \frac{d}{dt}(t^2|v_t|^2) + t^2\nu||v_t||^2 \leq c_2||v||^2(t^2|v_t|^2) + 2t|v_t|^2.
\end{eqnarray}
To get rid of the last term on the right hand side we add to (\ref{eqn:exp10}) equation (\ref{eqn:exp5}) multiplied by
$2t$. After simple calculations and using (\ref{eqn:exp6}) we obtain
\begin{eqnarray}\label{eqn:exp11}
    \frac{d}{dt}(t^2|v_t|^2 &+& t\nu||v||^2 + 2tj_\delta(v)) + t^2\nu||v_t||^2 \\
    &\leq& 2tc_1(||v_t|| + ||v|||v_t|^{1/2}||v_t||^{1/2} + |v|^{1/2}||v||^{1/2}|v_t|^{1/2}||v_t||^{1/2}) \nonumber \\
    &+& c_2||v||^2(t^2|v_t|^2) + \nu||v||^2 + 2j_\delta(v) \nonumber \\
    &+& c_3|v|^{1/2}||v||^{3/2}(t|v_t|)^{1/2}(t||v_t||)^{1/2}.
\nonumber
\end{eqnarray}
Define $y = t^2|v_t|^2 + t\nu||v||^2 + 2tj_\delta(v)$.
Using the Young inequality to the right hand side of 
(\ref{eqn:exp11}) and observing that $j_\delta(v) \leq C(||v||^2+1)$ for $0<\delta \leq 1$, we obtain at last the inequality of the form
\begin{eqnarray*}\label{eqn:exp12}
    \frac{d}{dt}y(t) + \frac{t^2\nu}{2}||v_t||^2 \leq C_1(t)y(t) + C_2(t),
\end{eqnarray*}
where the coefficients $C_i(\cdot)$, $i=1,2$, are locally integrable and do not depend on $\delta$. They are also independent of the initial conditions for $v$ in a given bounded sets in $H$. 
This proves that the time derivative of solutions of the regularized problems is uniformly bounded with respect to $\delta$ in
$L^\infty(\eta, T; H)\cap L^2(\eta, T; V)$ for all intervals $[\eta, T]$, $0<\eta < T$. As a consequence, this property holds for all $\chi\in {\cal B}_l^1$. In view of the above considerations this ends the proof of the existence of an exponential attractor.
\end{proof}
\renewcommand{\theequation}{\arabic{section}.\arabic{equation}}
\setcounter{equation}{0}
\section{Conclusions and some open problems} \label{conclusions}

In this paper we proved the existence of the global attractor of a finite fractal dimension and also of an exponential attractor for a Navier-Stokes flow with Tresca's boundary condition appearing in the theory of lubrication. In the end we would like to mention some related problems.

First, there is a question of the existence of global and exponential attractors for other contact problems with subdifferentiable boundary conditions. As concerns exponential attractors, the main difficulty seems to be produced by assumption (A10). Usually (\ref{exponential-problem}) follows from some better property of the time derivative of the solution, e.g., it easily follows if we know that $u'\in L^q(0,T;X)$ for $T>0$ and some $q>1$.
However, such a property is not known to hold for many contact problems.
On the other hand, for numerous quasistatic and dynamical contact problems this property is naturally satisfied 
\cite{millor-sofonea-telega-2010, Duv72}, however, for the associated semigroups acting in some more regular phase spaces $Y$ of initial conditions (with $Y$ compactly embedded in $X$). In this situation it seems important to study further regularity properties of solutions of contact problems to have
   $u'\in L^q(0,T;X)$ for $T>0$ and some $q>1$ for solutions with initial data in $X$ (and not in $Y$).
   
For some visco-plastic flows (e.g., Bingham flows) governed by variational inequalities (however, only with homogeneous or periodic boundary conditions) the regularity problem in question was solved e.g. in \cite{Lady-Seregin-1995, Seregin-1994} and used to study the time asymptotics of solutions.

In this context, there are other important open problems, namely these of the time asymptotics of solutions of the Navier-Stokes, visco-plastic and other fluid models governed by evolution variational or even hemivariational inequalities (cf., e.g., \cite{Migorski-Ochal-2007}) that take into account involved boundary conditions coming from a variety of applications in mechanics. 

\vspace{0.3cm}

{\bf Acknowledgements.} The author would like to thank very much the unknown referee for helpful comments which allowed to improve the paper in some important points.



\addcontentsline{toc}{chapter}{Bibliographie}
\begin{thebibliography}{00} 

\bibitem{millor-sofonea-telega-2010} {\sc M. Shillor, M.Sofonea, J.J. Telega}, 
        {\em Models and Analysis of Quasistatic Contact: Variational Methods},
        Springer-Verlag, Berlin Heidelberg, 2010.     

\bibitem{robinson-2011-dim} {\sc J.C. Robinson},
        {\em Dimensions, Embeddings, and Attractors},
        Cambridge University Press, UK, 2011. 
        
\bibitem{robinson-2001-infty}  {\sc  J.C. Robinson},
        {\em Infnite-dimensional Dynamical Systems},
        Cambridge University Press, Cambridge, UK 2001.   
            

\bibitem{chep-vish-2002} {\sc V.V. Chepyzhov, M.I. Vishik},
        {\em Attractors for Equations of Mathematical Physics},
        AMS, Providence, RI, 2002.
        
\bibitem{cholewa-dlotko 2000}  {\sc J. Cholewa, T. D{\l}otko},
        {\em Global Attractors in Abstract Parabolic Problems},
        Cambridge University Press, Cambridge, UK, 2000.
        
\bibitem{hale-1988} {\sc J.K. Hale},
        {\em Asymptotic Behavior of Dissipative Systems},
        AMS, Providence, RI, 1988.

\bibitem{temam-infty} {\sc R. Temam}, {\em Infinite-Dimensional Dynamical Systems in Mechanics and Physics}, 2nd. ed.,
        Springer-Verlag, New York, 1997.
        
\bibitem{malek-necas-1996} {\sc J. M\'alek, J. Ne\v{c}as},
        {\em A finite-dimensional attractor for three-dimensional flow of incompressible fluids},
        J. Diff. Eqns., Vol. 127 (1996), 498--518.         
        
\bibitem{malek-prazak-2002} {\sc J. M\'alek, D. Pra\v{z}\'ak},
        {\em Large time behavior via the method of $l$-trajectories},
        J. Diff. Eqns., Vol. 181 (2002), 243--279.          
        
\bibitem{Feireisl-Prazak-2010} {\sc E. Feireisl, D. Pra\v{z}\'ak},
        {\em Asymptotic Behavior of Dynamical Systems in Fluid Mechanics},
        AIMS, 2010.          
        
\bibitem{Miranville-Zelik-2008} {\sc A. Miranville, S. Zelik},
        {\em Attractors for dissipative partial differential equations in bounded and unbounded domains},
        Evolutionary Equations (Handbook of Differential Equations
        vol IV) Amsterdam, Elsevier, North-Holland (2008), 103--200.        
        
\bibitem{segatti-zelik-2010} {\sc A. Segatti, S. Zelik},
        {\em Finite-dimensional global and exponential attractors
        for the reaction-diffusion problem with an obstacle potential},
        Nonlinearity, Vol. 22 (2009), 2733--2760.         
        
\bibitem{Duv72}
        {\sc G. Duvaut,  J.L. Lions},
        {\em Les In\'equations en M\'ecanique et en Physique},
        Dunod, Paris, 1972.                        
        
\bibitem{mb-gl-cambridge-2009} {\sc M. Boukrouche, G. {\L}ukaszewicz},
        {\em Shear flows and their attractors}, 
        pp. 1-27 in Partial Differential Equations and Fluid Mechanics, edited by Jos\'e L. Rodrigo and James C.  
        Robinson,  London Mathematical Society Lecture Note Series (No. 364), Cambridge University Press, UK 2009.        
        
\bibitem{BoLu1-04}  {\sc M. Boukrouche, G. {\L}ukaszewicz},
        {\em An upper bound on the attractor dimension of a 2D
        turbulent shear flow in lubrication theory},
        Nonlinear Analysis TMA, Vol. 59 (2004), 1077--1089.        
        
\bibitem{BoLu2-04}  {\sc M. Boukrouche, G. {\L}ukaszewicz},
        {\em An upper bound on the attractor dimension of a 2D turbulent
         shear flow with a free boundary condition},
         pp. 61--72 in Regularity and other aspects of the Navier-Stokes equations,
         Banach Center Publ., Vol.70, Polish Acad. Sci., Warszawa, 2005.         
        
\bibitem{BoLuR-05} {\sc M. Boukrouche, G. {\L}ukaszewicz, and J. Real},
     {\em On pullback attractors for a class of two-dimensional
          turbulent shear flows},
          International Journal of Engineering Science.
          Vol.44  (2006), 830--844.        
        
\bibitem{mb-gl-parabolic-2008} {\sc M. Boukrouche, G. {\L}ukaszewicz},
        {\em On the existence of pullback attractor for a two-dimensional shear flow with Tresca's boundary 
         condition},
        pp. 81--93 in Parabolic and Navier-Stokes Equations, Banach Center Publ., Vol. 81, Polish Acad. Sci., 
        Warszawa 2008.        
        
\bibitem{zhong-2006} {\sc C.K. Zhong, M.H. Yang, C.Y. Sun}, 
        {\em The existence of global attractors for the norm-to-weak continuous semigroup and application to the 
         nonlinear reaction-diffusion equations}, 
        J. Diff. Eqns., Vol. 223 (2006), 367--399.
        
\bibitem{papa-1985} {\sc  P.D. Panagiotopoulos}, 
        {\em Inequality Problems in Mechanics and Applications},  
        Birkh\"auser Verlag,  Basel, 1985.         
        
        
\bibitem{haslinger} {\sc J. Haslinger, I. Hlav\'a\v{c}ek, J. Ne\v{c}as},
        {\em Numerical methods for unilateral problems in solid mechanics},
    in Handbook of Numerical Analysis, Vol IV. Edited by P.G. Ciarlet and J.L.Lions, 1996.         
        
\bibitem{necas-1967} {\sc J. Ne\v{c}as}, 
        {\em Les M\'ethodes Directes en Th\'eorie des Equations Elliptiques},
        Mason, Paris, 1967.
        
\bibitem{Lady-Seregin-1995} {\sc O. Ladyzhenskaya, G. Seregin},
        {\em On semigroups generated by initial-boundary problems describing two-dimensional visco-plastic flows},
        Amer. Math. Soc. Transl. (2) Vol. 164, 1995.  
        
\bibitem{Seregin-1994} {\sc G. Seregin},
        {\em On a dynamical system generated by the two-dimensional equations of the motion of a Bingham fluid},
        Journal of Mathematical Sciences, Vol. 70, No. 3, 1994 (translated from Zapiski Nauchnykh Seminarov LOMI, SSSR,
        Vol. 188, 128--142, 1991.        
        
\bibitem{Migorski-Ochal-2007} {\sc S. Mig\'orski, A. Ochal},
        {\em Navier-Stokes models modeled by evolution hemivariational inequalities},
        Discrete and Continuous Dynamical Systems Supplement (2007), 731--740.        

\end{thebibliography}
\end{document}